\documentclass[aps,pra,superscriptaddress,reprint]{revtex4-2}
\usepackage{graphicx}
\usepackage{amssymb,amsthm,amsmath,braket,mathtools,overpic}
\usepackage[bookmarksnumbered,linktocpage,hypertexnames=false,colorlinks=true,linkcolor=blue,urlcolor=blue,citecolor=blue,anchorcolor=green,breaklinks=true,pdfusetitle]{hyperref}
\usepackage[capitalize,nameinlink]{cleveref}
\usepackage{chngcntr}
\usepackage{titlesec}

\usepackage[caption=false]{subfig}

\newtheorem{thm}{Theorem}\crefname{thm}{Theorem}{Theorems}
\newtheorem{lem}[thm]{Lemma}\crefname{lem}{Lemma}{Lemmas}
\theoremstyle{definition}
\crefname{def}{Definition}{Definitions}

\DeclarePairedDelimiter\abs{\lvert}{\rvert}
\DeclarePairedDelimiter\norm{\lVert}{\rVert}

\newcommand{\eps}{\varepsilon}
\newcommand{\microsec}{\mu\textrm{s}}

\titleformat{\section}[runin]{\normalfont\bfseries\itshape}{\thesection}{1em}{}[---]
\titlespacing*{\section}{\parindent}{\baselineskip}{0pt}

\titleformat{\subsection}[runin]{\normalfont\itshape}{\thesubsection}{1em}{}[---]
\titlespacing*{\subsection}{\parindent}{\baselineskip}{0pt}

\begin{document}

\title{Adiabatic preparation of many-body quantum states: getting the beginning and ending right}
\author{Emil T. M. Pedersen}
\affiliation{Department of Mathematical Sciences, University of Copenhagen, Universitetsparken 5, 2100 Copenhagen, Denmark}
\author{Freek Witteveen}
\affiliation{Centrum Wiskunde \& Informatica and QuSoft, Science Park 123, 1098 XG Amsterdam, the Netherlands}
\author{Klaus M{\o}lmer}
\affiliation{Niels Bohr Institute, University of Copenhagen, Jagtvej 155, 2200 Copenhagen, Denmark}
\author{Matthias Christandl}
\affiliation{Department of Mathematical Sciences, University of Copenhagen, Universitetsparken 5, 2100 Copenhagen, Denmark}

\begin{abstract}
We present numerical calculations, and simulations performed on a Rydberg atom quantum simulator, of the adiabatic evolution of many-body quantum systems around a quantum phase transition. We demonstrate that the end-to-end transfer error, for a given process duration and dissipative losses, can be suppressed by adopting smooth initial and final scheduling functions for the Hamiltonian. We consider a one-dimensional mixed-field Ising model, as well as a chain of Rydberg atoms, and compare numerical calculations and experimental results for the end-to-end transfer error with different schedule functions. We show, in particular, that if the time dependent Hamiltonian is $n$ times differentiable with vanishing $1^{st}$ to $n^{th}$ order derivatives in the beginning and in the end, the infidelity with respect to the final adiabatic eigenstate scales as $1/T^{n+1}$ when evolving for time $T$. 
\end{abstract}

\maketitle

The adiabatic theorem~\cite{Born} states that a system prepared in an eigenstate of a Hamiltonian $H$ and evolved under a slowly varying Hamiltonian $H(t)$ remains in the corresponding instantaneous eigenstate, up to small non-adiabatic corrections. Adiabatic evolution is therefore widely used for preparing eigenstates of complex many-body Hamiltonians. Several strategies are known to suppress non-adiabatic transitions, including locally slowing the evolution near small gaps~\cite{Avron}, choosing favorable paths in parameter space~\cite{Schiffer}, and adding counter-diabatic terms that cancel diabatic couplings~\cite{Muga}. The latter yields a `shortcut to adiabaticity', where the state is subject to a \emph{modified} Hamiltonian, in order to follow the eigenstate of the original target Hamiltonian. These techniques can be combined with variational and optimal-control methods~\cite{Motzoi,Opatrny,caneva2009optimal,brif2014exploring}.

Here we focus on a different aspect of adiabatic evolution that has been known since early work: the fidelity of the \emph{final} state can be much higher than the fidelity of the instantaneous state during the evolution. In the Landau--Zener problem~\cite{landau1932theorie,zener1932non,Majorana}, a constant coupling $V$ couples two levels with linearly varying energy difference $\Delta E = \alpha t$. The leading-order non-adiabatic error near the avoided crossing scales as $\alpha/V$, consistent with the Born--Fock estimates. However, at large positive times the loss of population from the adiabatic state is exponentially small, $\propto \exp(-V/\alpha)$~\cite{dykhne1960quantum,davis1975nonadiabatic}. Similar exponential suppression appears in other analytically solvable two-level problems and can be understood via superadiabatic basis transformations that absorb low-order non-adiabatic corrections~\cite{berry1993universal}. Crucial to this mechanism is the requirement that the adiabatic eigenstates connect smoothly to the initial and final state, i.e., that the time dependent Hamiltonian evolves in a smooth manner throughout the entire process. Additionally, the derivatives need to vanish at the beginning and end of the Hamiltonian path.

In this work, we investigate the precise relation between the smoothness of the boundary conditions and adiabatic state preparation in many-body systems. We combine three components.
We prove a rigorous adiabatic theorem showing that $n$ vanishing time derivatives of the Hamiltonian at the initial and final times give an error bound of $O(\eps^{n+1})$ with $\eps = 1/T$, improving previous rigorous estimates of $O(\eps^n)$.
We then study this numerically for a one-dimensional mixed-field Ising model, where we construct schedule functions that control the number of vanishing derivatives at the boundaries while leaving the bulk dynamics largely unchanged.
Finally, we study this numerically and experimentally on a Rydberg atom chain implemented on the Aquila neutral-atom quantum simulator~\cite{wurtz2023aquila}, where we test these schedule constructions under noise and measurement errors.
Our main conclusion is that enforcing smoothness at the beginning and end of the protocol reduces end-to-end infidelity, and that this improvement can be obtained with minimal modification of a given schedule.

\section{Adiabatic theorem}\label{sec:adiabatic theorem}
We recall the setting of adiabatic Hamiltonian evolution.
Consider an $n$ times continuously differentiable family of Hamiltonians $H(\tau)$, where $0\leq \tau\leq 1$.
We consider time evolution under the Hamiltonian $H(t/T)$ for time $t=0$ to $t=T$ and write
\begin{align*}
    \eps = T^{-1}, \qquad \tau = t/T.
\end{align*}
Let $\ket{\Phi(\tau)}$ be an eigenstate of $H(\tau)$ with eigenvalue $E(\tau)$, separated from the rest of the spectrum by a finite gap (so there are no level crossings and the state and energy are continuous functions of $\tau$). A true state $\ket{\psi(\tau)}$ evolves according to the Schr\"odinger equation with initial condition $\ket{\psi(0)} = \ket{\Phi(0)}$. For $\eps$ small, the adiabatic theorem guarantees that $\ket{\psi(1)}$ has large overlap with $\ket{\Phi(1)}$.

The error depends, amongst other things, on the number of vanishing derivatives of $H(\tau)$ at $\tau=0$ and $\tau=1$. If $H(\tau)$ has $n$ vanishing derivatives at both boundaries, previous rigorous results show an error bound of order $O(\eps^n)$~\cite{garrido1962degree,lidar2009adiabatic,jansen2007bounds,hagedorn2002elementary,campos2018error}. We extend these results as follows.

\begin{thm}\label{thm:adiabatic derivatives}
    Consider a Hamiltonian $H(\tau)$ for $\tau \in [0,1]$ that is $n$ times differentiable in $\tau$, with the first $n$ derivatives vanishing at $\tau=0$ and $\tau=1$. 
     Let $E(\tau)$ be an eigenvalue of $H(\tau)$ separated from the rest of the spectrum by a finite gap, and let $\ket{\Phi(\tau)}$ be the corresponding eigenstate. Denote by $\ket{\psi(\tau)}$ the state obtained by evolving from $\ket{\psi(0)} = \ket{\Phi(0)}$ under $H(t/T)$ with $\eps = 1/T$. Then
    \begin{align*}
        \norm{\ket{\psi(1)}-e^{-i\int_0^1 E(\tau')d\tau'/\eps}\ket{\Phi(1)}} \leq O(\eps^{n+1})
    \end{align*}
\end{thm}
The proof, given in \cref{sec:proof}, is based on the superadiabatic expansion \cite{berry1993universal,hagedorn2002elementary,lidar2009adiabatic}; the leading order error term for small $\eps$ comes from the boundary component of this expansion.
The scaling $O(\eps^{n+1})$ dominates only for sufficiently small $\eps$. For larger $\eps$ the adiabatic error decays approximately as $\exp(-c/\eps)$, where $c$ depends on the analyticity properties of $H(\tau)$ and on the minimal spectral gap between $E(\tau)$ and the rest of the spectrum~\cite{lidar2009adiabatic,hagedorn2002elementary}. We refer to the range of $\eps$ where the polynomial behavior dominates as the \emph{polynomial regime}, and to the range where the exponential dependence dominates as the \emph{exponential regime}. In the following sections, we explore this numerically in concrete many-body models.

\section{Ising model}\label{sec:Ising}

We consider a one-dimensional mixed-field Ising chain with nearest-neighbor couplings,
\begin{equation}\label{eq:MFIM}
    H = \sum_{i=1}^{L-1}S_i^z S_{i+1}^z + \sum_{i=1}^{L} (gS_i^x+hS_i^z),
\end{equation}
with spin-$\tfrac{1}{2}$ operators $S^x_i$ and $S_i^z$. \cref{fig:schedule} shows the energy gap between the ground and lowest excited state as a function of $g$ and $h$ for $L = 21$. For $g=0$ and $h<-1$, the ground state is the ferromagnetic product state $\ket{\Psi_{\textrm{initial}}}=\ket{0}^{\otimes L}$, while for $g = 0$ and $-1 < h < 0$ it is antiferromagnetic, $\ket{\Psi_{\mathrm{target}}} = \ket{0101\ldots 0}$. We use adiabatic evolution to transform $\ket{\Psi_{\mathrm{initial}}}$ into $\ket{\Psi_{\mathrm{target}}}$.

\begin{figure}[t]
  \centering
  \begin{overpic}[width=0.95\linewidth,grid=false]{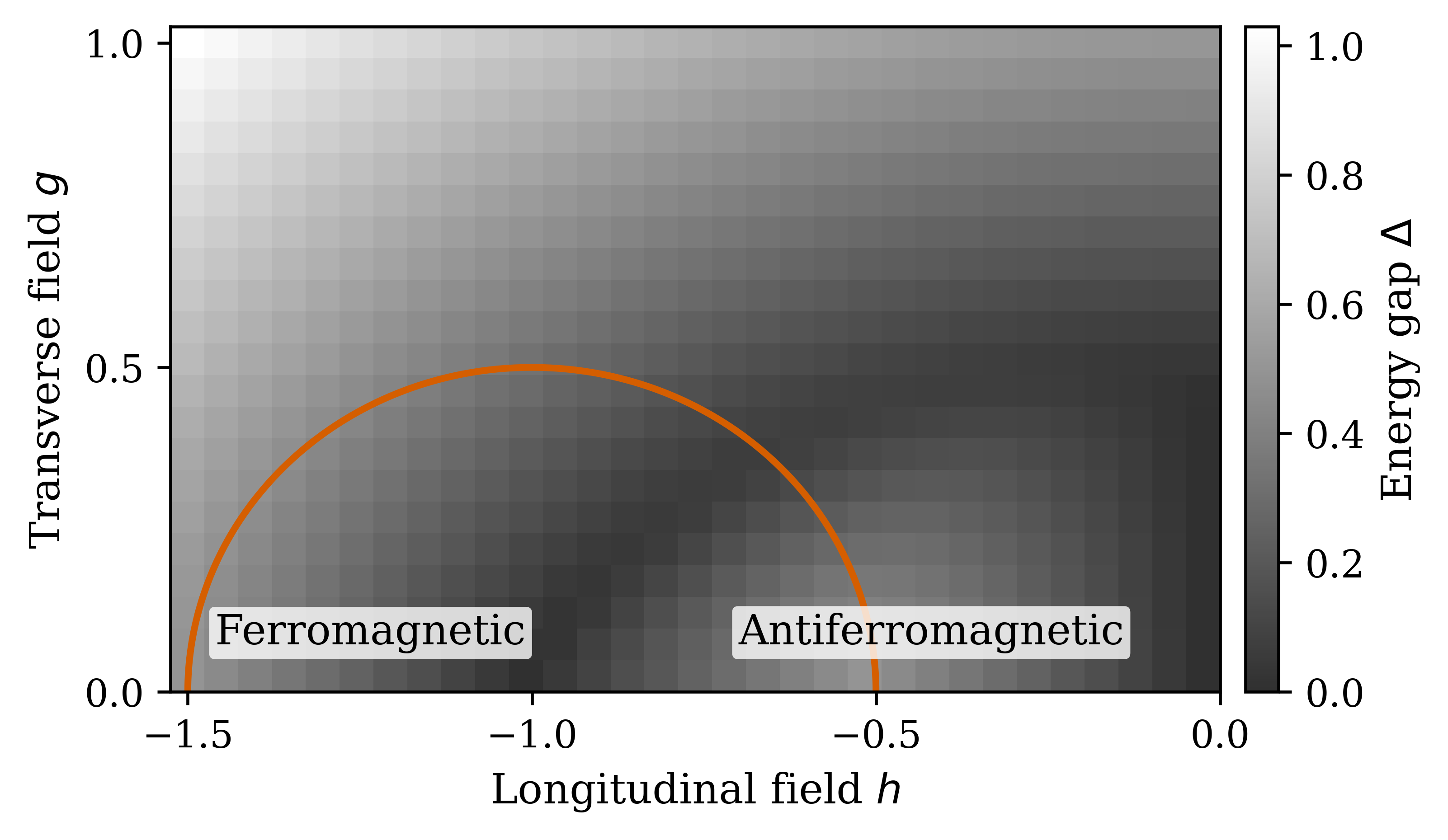}
    \put(-6,50){(a)}
  \end{overpic}
  \hfill
  \begin{overpic}[width=\linewidth,grid=false]{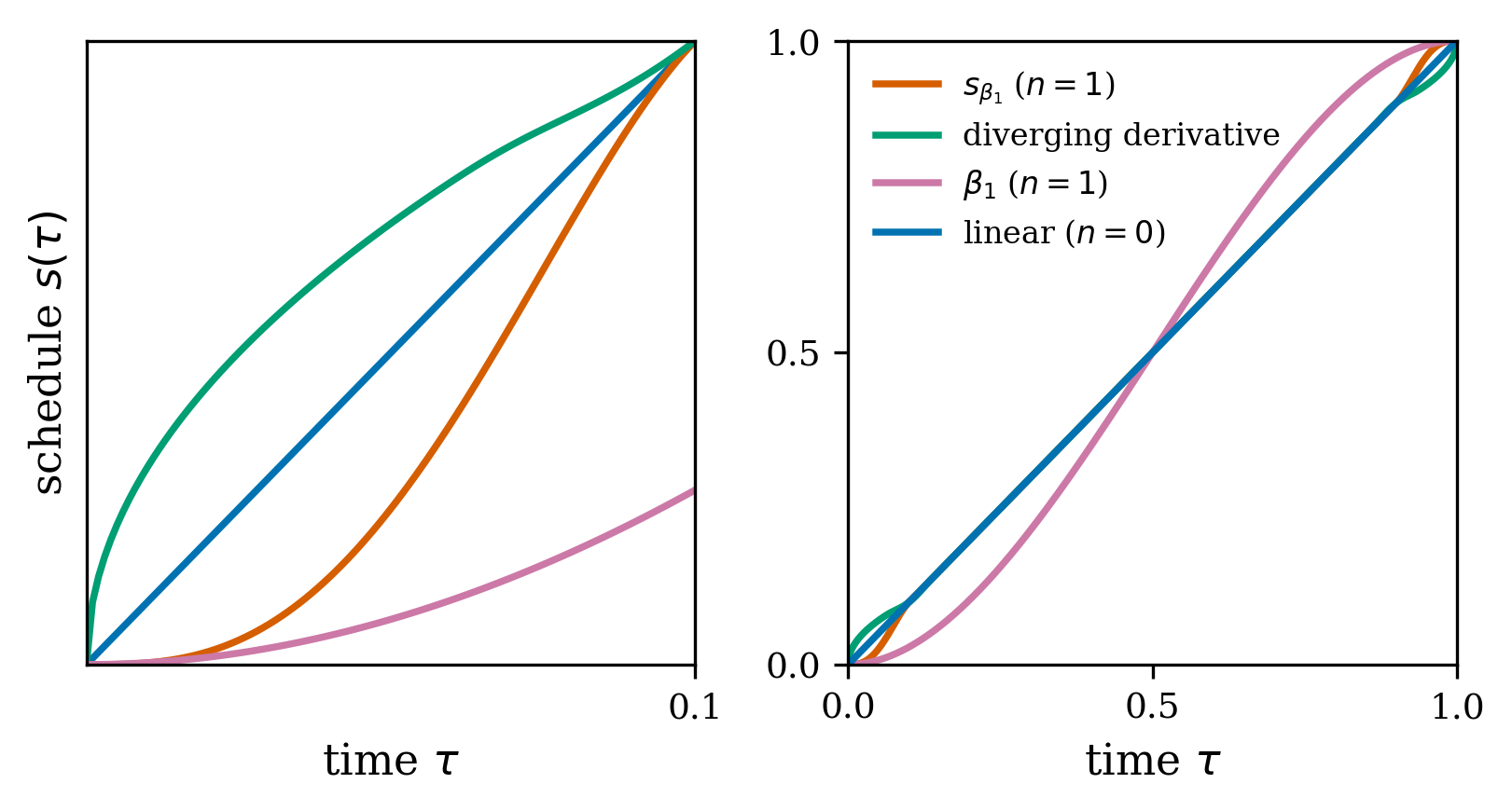}
    \put(-3,47){(b)}
  \end{overpic}
  \caption{(a) Energy gap and phases of the Ising model \eqref{eq:MFIM} with $L=21$ spins. The curve shows the path \eqref{eq:path} from the ferromagnetic to the anti-ferromagnetic phase. (b)~Schedule functions with linear, diverging and  vanishing $n^{th}$ order derivative in $\tau=0,1$ - see detailed forms in \cref{sec:schedule}. On the right the full schedule on $[0,1]$, on the left just the start.}
  \label{fig:schedule}
\end{figure}

\subsection{Schedule functions for adiabatic evolution}
We first consider the case where the matrix norm of $\partial_\tau H(\tau)$ is constant, and we assume a semicircular path parameterized by
 \begin{equation}\label{eq:path}
    \begin{aligned}
        h(\tau)=-(1 + \frac12 \cos(\pi \tau)), \,  g(\tau)=\frac12 \sin(\pi \tau)
    \end{aligned}
 \end{equation}
for $0\leq \tau \leq 1$. 

Given a Hamiltonian path $H(\tau)$ we can change the speed along the path, by switching to $H(s(\tau))$, where $s:[0,1]\rightarrow[0,1]$ is a monotonously increasing function from $s(0)=0$ to $s(1)=1$ which we will call the schedule. We are here not looking for an optimal schedule but we are interested in a comparison between the case of vanishing and non-vanishing boundary derivatives.
The error of the adiabatic evolution is determined by two factors: the speed at which the path passes through a phase transition, and the value of the derivatives at the end points. 
The reference schedule is a linear schedule.
In order to study the effect of the boundary conditions, we propose schedule functions $s_{\beta_n}$ which employ beta functions to set a fixed number $n$ of derivatives to zero at the beginning and end, but do not change the speed in the middle. One can also use a beta function $\beta_n$, that does change the speed in the bulk, as the schedule. We also consider a schedule with square root scaling, and hence with diverging derivatives at the boundaries.
Examples are shown in \cref{fig:schedule} and described in detail in \cref{sec:schedule}.

\subsection{Numerical results}\label{sec:results ising}
\begin{figure}[t]
    \centering
    \includegraphics[width=0.95\linewidth]{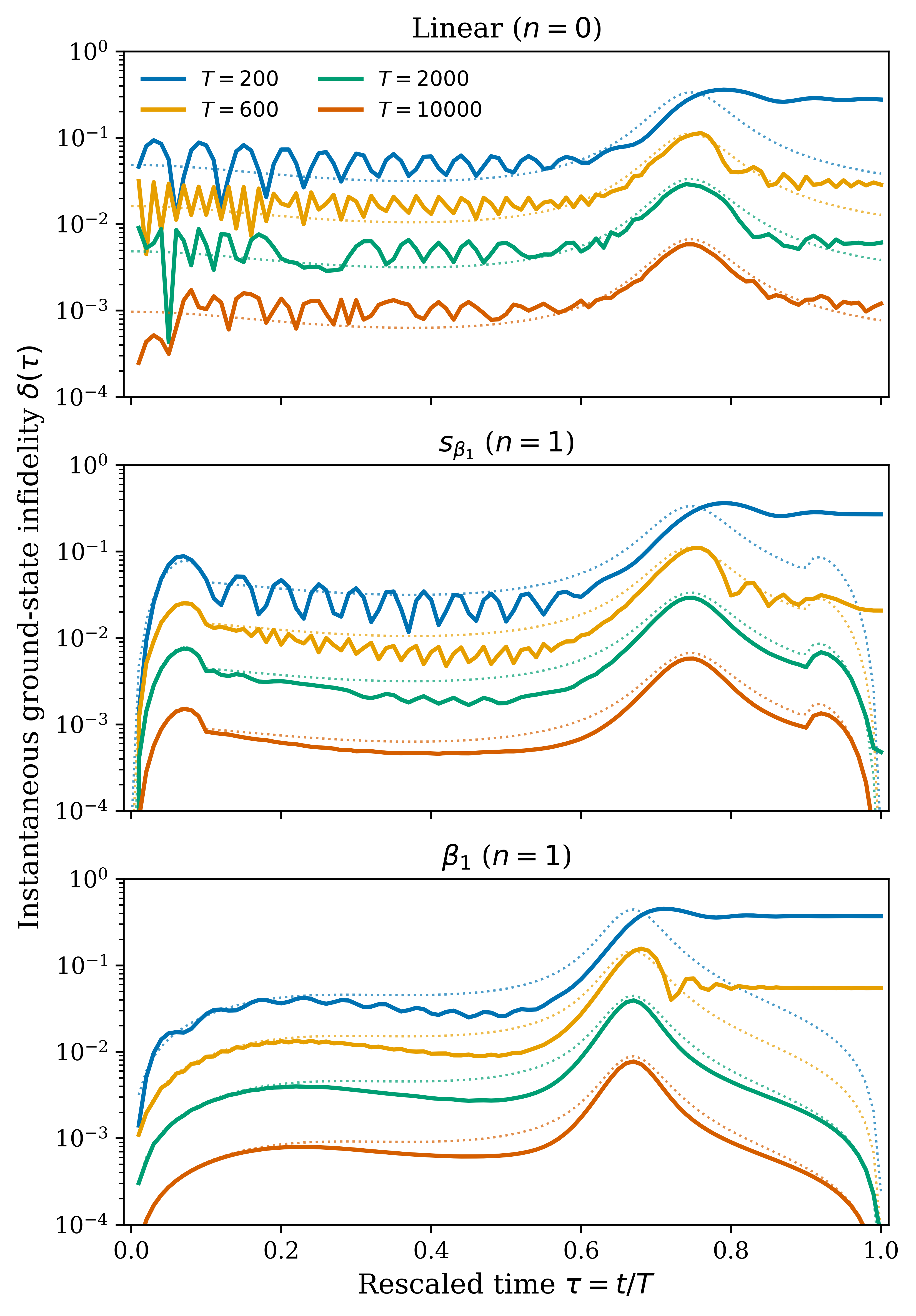}
    \caption{Time-evolution of the infidelity $\delta(\tau)$ during adiabatic passage of a length $L=11$ Ising chain for the linear, $s_{\beta_1}$, and $\beta_1$ schedules as shown in \cref{fig:schedule}. Note that the infidelity peak is slightly higher and occurs at an earlier time for $\beta_1$ because $\beta_1$ passes through the phase transition faster and at an earlier value of $\tau$. The color represents the total time $T$. Solid lines are the numerically simulated results. Dotted lines show the first-order component $\approx\eps\gamma_0(\tau)/\Delta_{01}(\tau)$.}\label{fig:infidelity_time_evolution}
\end{figure}

We simulate the dynamics of the mixed-field Ising model using matrix product states~\cite{hauschild2018tenpy}.
First we numerically study the one-dimensional mixed-field Ising model \eqref{eq:MFIM} with antiferromagnetic nearest-neighbor interactions. We use the Hamiltonian path and the schedules shown in \cref{fig:schedule} to prepare an antiferromagnetic state from a ferromagnetic state. \cref{fig:infidelity_time_evolution} shows the time-evolution of the instantaneous infidelity $\delta(\tau)=\sqrt{1-|\braket{\psi(\tau)|\Phi(\tau)}|^2}$ with the ground state $\ket{\Phi(\tau)}$ of $H(\tau)$, for different $\eps=1/T$. We compare three schedules: linear, $s_{\beta_1,f}$, and $\beta_1$.
All three schedules have similar infidelity values at intermediate times and a comparable peak value of infidelity at the minimum gap. 
We also numerically compare the first order correction, as derived in \cref{sec:first-order}; it shows good qualitative agreement for small $\eps$.
Very fast passage, corresponding to the exponential regime, results in the infidelity not recovering after the minimum gap, whereas slower passage, corresponding to the polynomial regime, recovers after the minimum gap. The difference between the schedules lies mainly at the beginning and end of the passage. For the linear schedule, the infidelity instantly jumps up and never fully recovers below this $\mathcal O(\eps)$ value at the end of the evolution. In contrast, for vanishing boundary derivatives, infidelity increases smoothly at the beginning and decreases smoothly at the end (in the polynomial regime), as the first-order term vanishes again. We also note that the initial jump is accompanied by subsequent oscillations of the infidelity throughout the dynamics of the linear schedule and the $s_{\beta_1}$ schedule.

\begin{figure}
    \includegraphics[width=\linewidth]{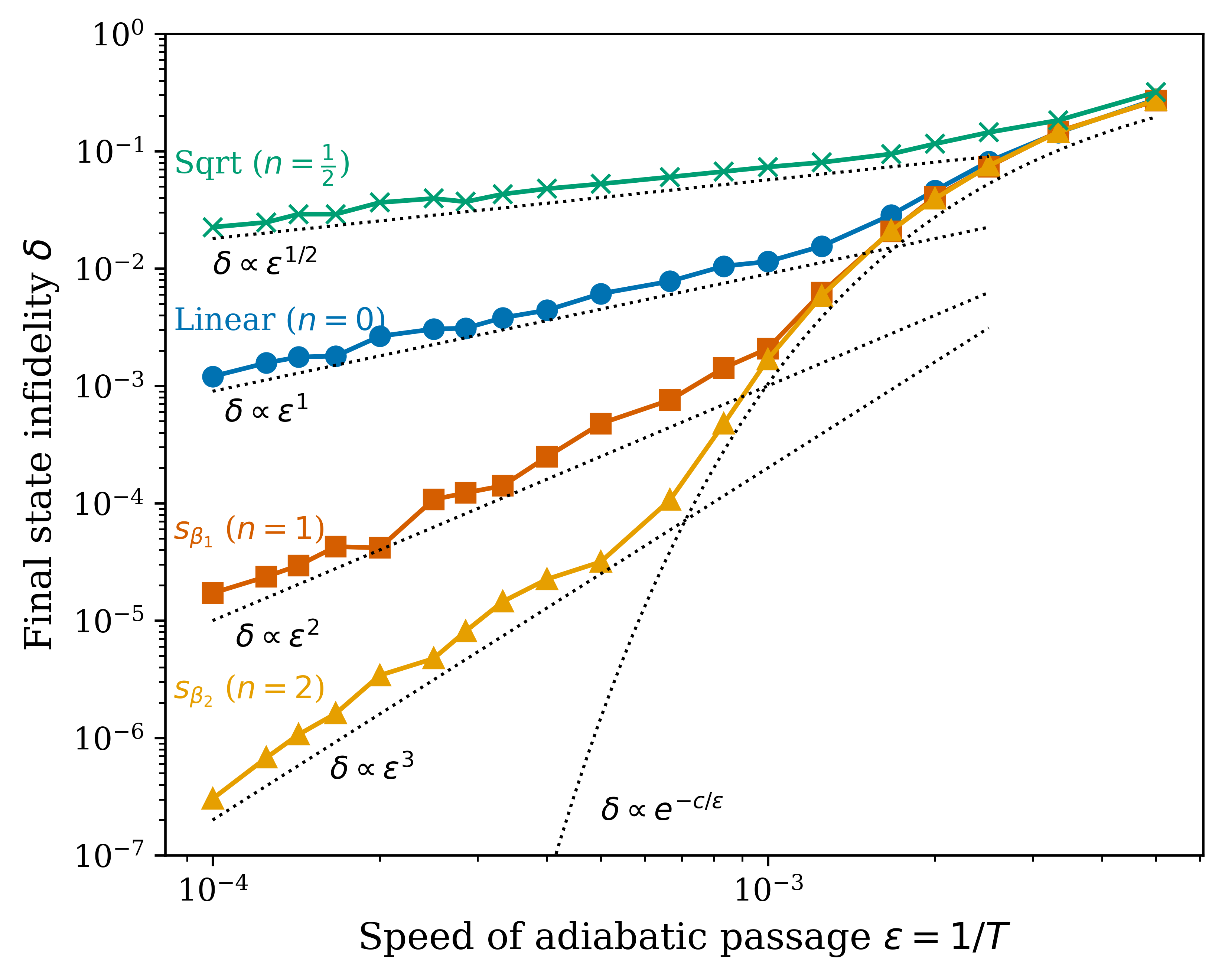}
    \caption{Ising model log-log plot of the final infidelity as a function of $\eps = 1/T$ for a length $L=11$ chain. In black are the theoretical scalings for the polynomial and exponential regime, with prefactors fitted to numerical results. Results are shown for $s_{\beta_n}$ schedules with $n$ vanishing derivatives for $n=0,1,2$ and for the Sqrt-schedule.}\label{fig:final_fidelity_L11}
\end{figure}

\cref{fig:final_fidelity_L11} shows the final infidelity $\delta = \sqrt{1-| \braket{\psi(1)|\Psi_{\textrm{target}}}|^2}$ as a function of $\eps=\frac{1}{T}$ for different schedules and chain length $L=11$. The black lines show the predicted scalings $\delta\propto\eps^{n+1}$, where $n$ is the number of vanishing derivatives, in the polynomial regime and $\delta\propto e^{-\frac{c}{\eps}}$ in the exponential regime. The prefactors and exponent $c$ were chosen to fit the numerical data. While we are not aware of theoretical results for the polynomial regime of the diverging derivative case, the numerical results match $\delta\propto\eps^{\frac{1}{2}}$.

To study the dependence on system size,
\cref{fig:final_fidelity_directBeta} shows the final infidelity as a function of $\eps=\frac{1}{T}$ for different schedules, but now for chain lengths $L=11$ as well as $L=21$.
We obtained similar results for all odd lengths between $L=11$ and $L=21$.
We again distinguish an exponential regime and a polynomial regime.
The exponential regime depends on system size: 
larger lengths $L$ take successively longer to reach the polynomial regime, and we fitted the exponent in the exponential regime to a power law dependence $c\approx 0.52L^{-1.83}$. There is no significant dependence on $n$, the number of vanishing derivatives, for the $s_{\beta_n}$ schedules in the exponential regime.
On the other hand, in the polynomial regime, the infidelity is governed by the value of  $n$, instead of $L$.
\cref{fig:final_fidelity_directBeta} also shows the results when using the $\beta_n$ schedule. We see qualitatively similar transitions from exponential to $\eps^{n+1}$ scaling of the infidelity, but the infidelity is worse at large $\eps$ with slower exponential convergence. Unlike the $s_{\beta_n}$ schedules, the exponent $c$ increases significantly with $n$ due to faster passage through the minimum gap. For small enough $\eps$ the $\beta_n$ schedules eventually catch up with and outperform $s_{\beta_n}$, reaching a $\eps^{n+1}$ regime with a better prefactor. We thus see a trade-off between performance at short and long times. For very small $\eps$ the $\beta_n$ schedules give the best results, while at intermediate $\eps$ the $s_{\beta_n}$ schedules give an infidelity that is up to an order of magnitude better.

\begin{figure}
    \includegraphics[width=0.95\linewidth]{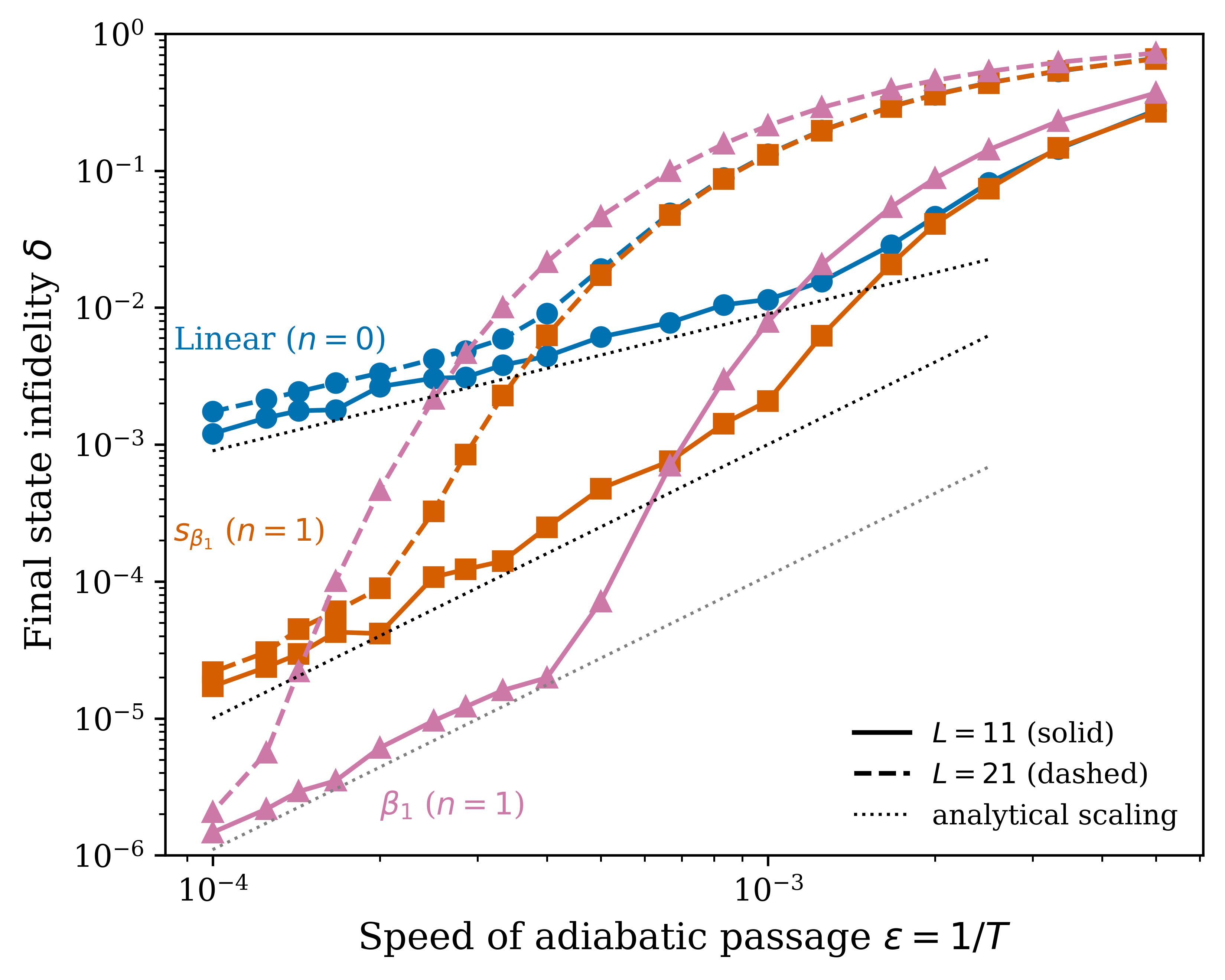}
    \caption{Log-log plot of final infidelity as a function of $\eps$, similar to \cref{fig:final_fidelity_L11}, for system size $L=11$ and $L=21$, to show dependence on system size.
    We show the $n=0$ schedule and the $\beta_n$ and $s_{\beta_n}$ schedule for $n=1$ (which have the same asymptotic scaling in $\eps$, but with a different prefactor).}\label{fig:final_fidelity_directBeta}
\end{figure}

\section{Results on a quantum simulator}\label{sec:quantumsimulator}

The above numerical results demonstrate that controlling the vanishing boundary derivatives is a powerful technique and can improve the fidelity of adiabatic state preparation by orders of magnitude, in an ideal noiseless setting.
We decided to run the protocol on the Aquila neutral atom device to investigate its performance on real (noisy) hardware.
The neutral atoms are governed by the Hamiltonian \cite{wurtz2023aquila}
\begin{equation}\label{eq:Haquila}
\begin{aligned}
    \frac{H}{\hbar} &= \frac{\Omega}{2}\sum_i \left(e^{i\phi}\ket{0_i}\hspace{-0.5ex}\bra{1_i} + h.c.\right) - \Delta \sum_i n_i\\
    &\qquad+ \sum_{i<j}n_i n_j\frac{C_6}{|\vec{r}_i-\vec{r}_j|^6}
\end{aligned}
\end{equation}
where $\ket{0_i},\ \ket{1_i}$ denote the ground and a Rydberg excited state of atom $i$, $\Omega$ is the Rabi frequency of coherent excitation of the atoms, and $n_i=\ket{1_i}\bra{1_i}$ and $C_6 = 862690 \cdot 2\pi \cdot \textrm{MHz}\cdot \textrm{\textmu s}^{-6}$ are the projection operators on Rydberg states and the strength of the Van der Waals interaction between Rydberg excited atoms.

This Hamiltonian does not implement an exact mixed field Ising model, since the nearest neighbor interaction is replaced by an all-to-all van der Waals interaction, but it displays a qualitatively similar phase transition from a ferromagnetic to an antiferromagnetic ground state.
We consider a spin chain with constant spacing $a=5.6\textrm{\textmu m}$. As with the Ising model, we consider adiabatic passage from the ferromagnetic to antiferromagnetic ground state, now along an elliptical path to accommodate constraints on the maximum Rabi frequency $\Omega$ in experiments. We use the constant speed parametrization
\begin{equation}\label{eq:Rydpath}
    \begin{aligned}
        & \Delta(\tau)=\Delta_R\cos\left(E^{-1}\left(\frac{P}{2\Omega_R}\tau,m\right)\right) \\
        & \Omega(\tau)=\Omega_R \sin\left(E^{-1}\left(\frac{P}{2\Omega_R}\tau,m\right)\right),
    \end{aligned}
 \end{equation}
where $\Omega_R=2.5\cdot 2\pi \cdot\textrm{MHz}$ is the maximum Rabi frequency and $\Delta_R=8.75\cdot 2\pi\cdot\textrm{MHz}$ and $\tau\in[0,1]$. Here $P=4\Delta_R E(\frac{\pi}{2},e^2)$ is the perimeter, $m=\frac{e^2}{e^2-1}$ and $e=\sqrt{1-\frac{\Omega_R^2}{\Delta_R^2}}$ the eccentricity, where 
 \begin{equation}
     E(\phi,m)=\int_0^{\phi}\sqrt{1-m\sin^2(t)}dt
 \end{equation}
 is the elliptic integral of the second kind, and we define the inverse elliptic integral of the second kind $E^{-1}(u,m)$ as being equal to $\phi_u$ such that $E(\phi_u,m)=u$.

\begin{figure}[t]
  \centering
  \begin{overpic}[width=0.9\linewidth,grid=false]{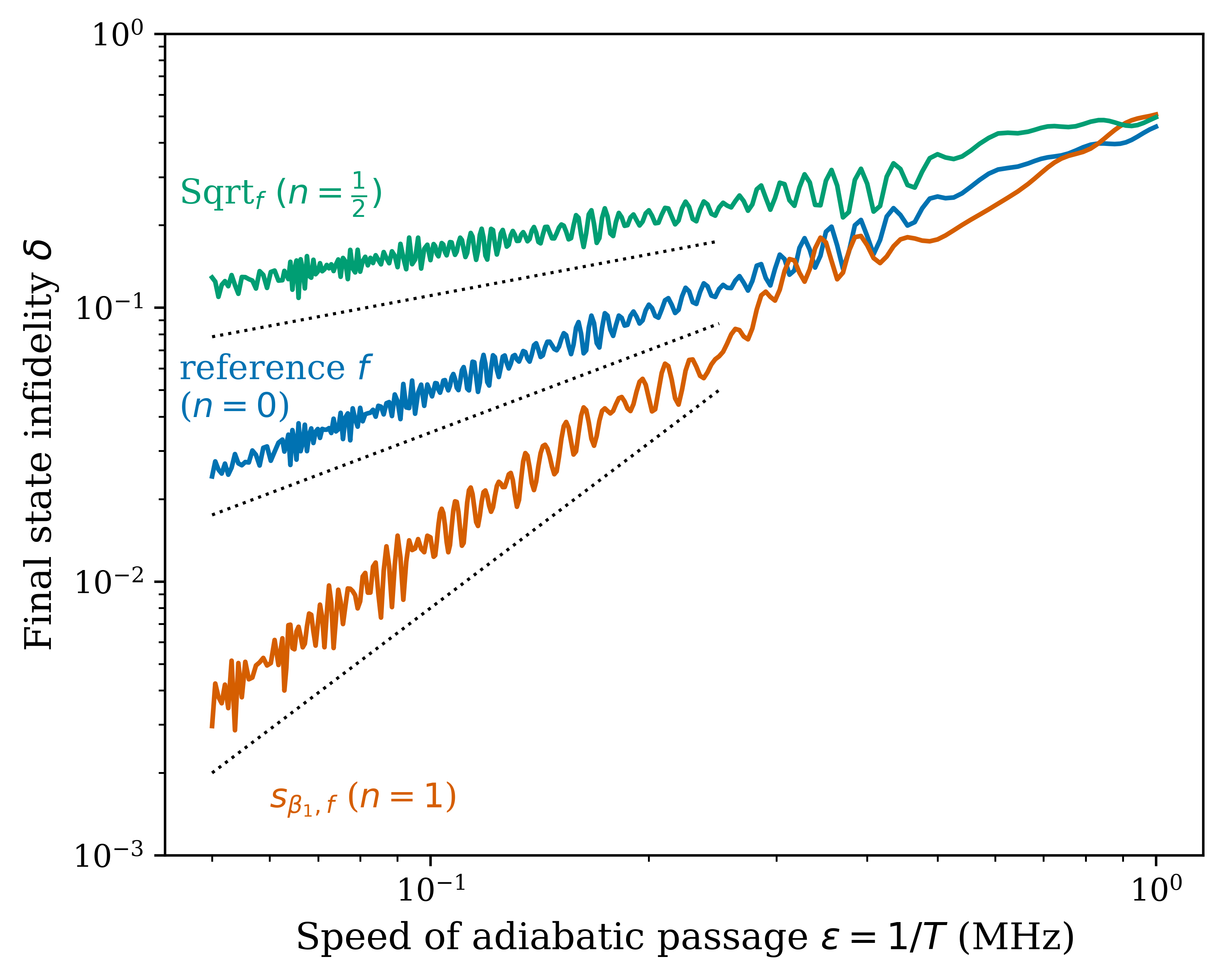}
    \put(0,70){(a)}
  \end{overpic}
  \hfill
  \begin{overpic}[width=0.9\linewidth,grid=false]{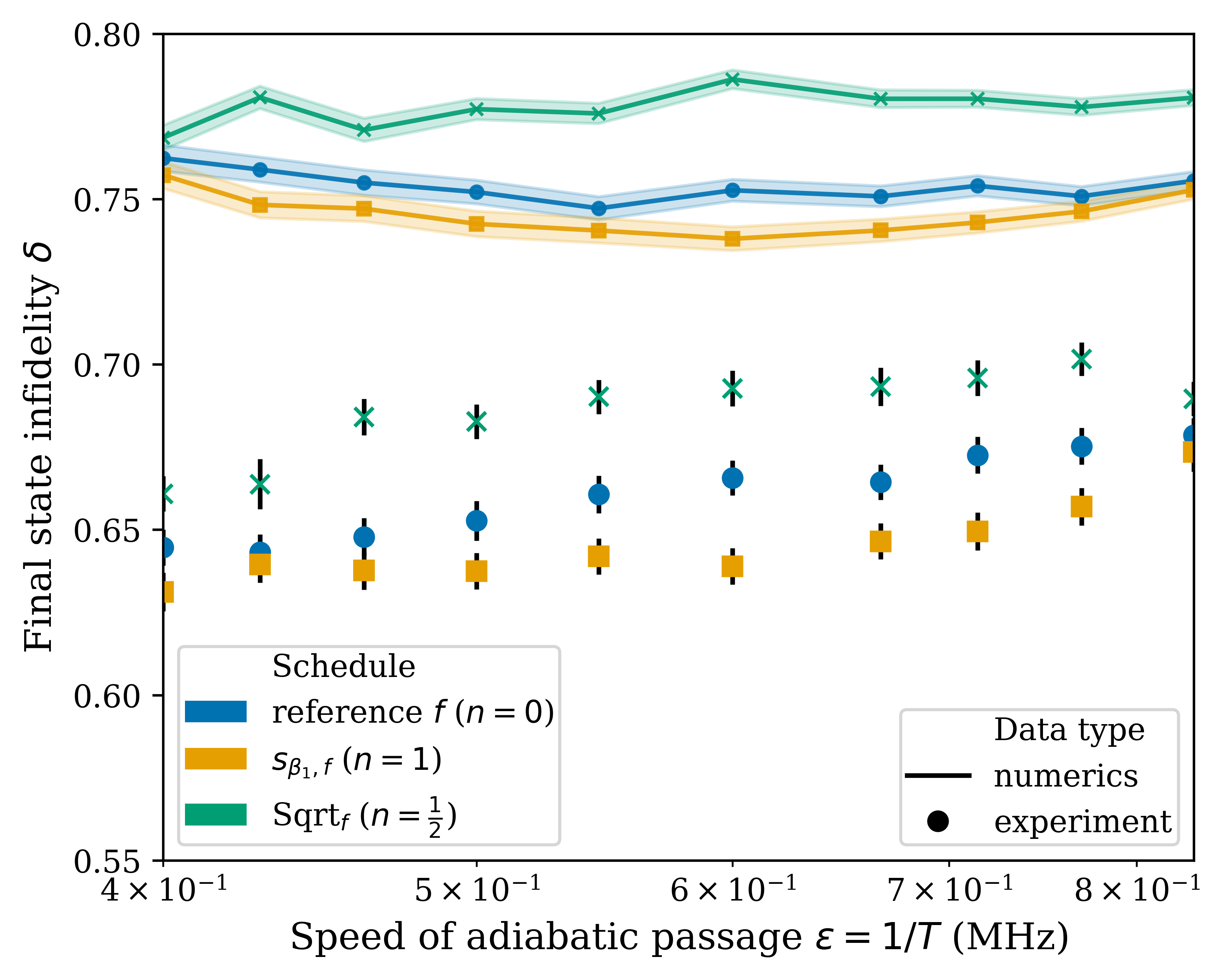}
    \put(0,70){(b)}
  \end{overpic}
  \caption{Rydberg model plot of final infidelity as a function of $\eps=\frac{1}{T}$. (a) Numerical results for a (noiseless) simulation, with a schedule derived from \cref{eq:IGapICosP05} and \cref{eq:igap_schedule}, for a reference schedule with no vanishing derivatives, the $s_{\beta_1 , f}$ schedule with $n=1$ vanishing derivatives, and the Sqrt$_f$-schedule with diverging derivatives. (b) Noisy results. Here we consider a smaller range of $T=1.2$\textmu{s} to $T=2.5$\textmu{s}. Results from a numerical simulation, using the error model of \cite{wurtz2023aquila}, as well as experiments on Aquila.}
  \label{fig:final_infidelity_Aquila_1}
\end{figure}

Using the trivial schedule $s(\tau) = \tau$ and neglecting noise, we again observe the expected improvement from vanishing derivatives, but only at times and infidelities beyond the practical range of the Aquila device.
To access a regime where boundary effects are visible at shorter times, we use a reference schedule that slows down near the minimum many-body gap, described in \cref{sec:schedule}. Similarly to before, we use this to construct schedules with vanishing boundary derivatives.
\cref{fig:final_infidelity_Aquila_1} shows the numerical results with this schedule for a chain of length 11 without noise, with different smoothness at $\tau = 0,1$.
We also used the built-in noise model of the Bloqade library and the characterization in \cite{wurtz2023aquila} to numerically simulate the results for this schedule using the Monte Carlo wavefunction method. These results are shown in \cref{fig:final_infidelity_Aquila_1} for a chain of length 11, together with the results from the Aquila device. We note that all infidelities are greater than $\approx 0.62$ due to systematic measurement errors \cite{wurtz2023aquila}. Instead of observing an infidelity that decreases with total time $T$, decoherence takes over and causes the infidelity to increase beyond $T\approx2\microsec$, and beyond $T\approx 4$ the difference between schedules with a different number of vanishing boundary derivatives becomes negligible.
In the interval where the infidelity is minimal, we still observe a separation between schedules with different numbers of vanishing boundary derivatives, albeit less pronounced than in the noiseless case.
The noise model from \cite{wurtz2023aquila} overestimates the infidelity. This reflects improvements in the hardware; the experiments shown in \cref{fig:final_infidelity_Aquila_1} were performed in June and July 2025.

\section{Conclusion and outlook}
We have presented a theoretical and experimental study of adiabatic state preparation in many-body systems with a focus on the role of boundary smoothness in time. The main theoretical result is an adiabatic theorem showing that $n$ vanishing time derivatives of the Hamiltonian at the initial and final times lead to a preparation error scaling as $T^{-(n+1)}$ for closed systems. We constructed schedule functions that impose these boundary conditions while minimally modifying a given reference schedule at intermediate times.
Numerical simulations in the Ising model and Rydberg model show that the vanishing boundary derivatives reduce the final infidelity by orders of magnitude in the polynomial regime in a noiseless setting, without influencing the results for fast passage in the exponential regime. 
For fast passage there is no advantage in using vanishing boundary derivatives and one is better off using other techniques such as passing slowly through the gap. What techniques to use thus depends heavily on the target timescale. Our experiments on the Aquila quantum processor show that vanishing boundary derivatives can have measurable signature on noisy experimental hardware. The gains are smaller than in the ideal noiseless setting, but come at little to no cost.  Our construction has the advantage of being flexible and easy to combine with other established techniques, such as passing slowly through the minimal gap, so there are plenty of use cases to explore.

\section{Acknowledgments}
The work was supported by the European Research Council (ERC Grant Agreement No. 81876), the Novo Nordisk Foundation (Grant NNF20OC0059939 “Quantum for Life”), and VILLUM FONDEN  via the QMATH Centre of Excellence (Grant No. 10059).
Numerical results for the Ising model were obtained with the Density matrix renormalization group (DMRG) and time-evolving block decimation (TEBD) algorithms implemented by the TeNPy library \cite{hauschild2018tenpy}. The Bloqade library \cite{bloqade2023quera} was used for numerical simulation of the Rydberg model.

\bibliographystyle{unsrt}
\bibliography{library}

\newcounter{eqsave}
\setcounter{eqsave}{\value{equation}}
\appendix
\counterwithout{equation}{section}
\renewcommand\theequation{\arabic{equation}}
\setcounter{equation}{\value{eqsave}}

\section{Schedule functions}\label{sec:schedule}

The linear ramp $s(\tau)=\tau$ will serve as a reference schedule. We construct functions with $n$ vanishing boundary derivatives based on $\beta_n(x)=B(x,n+1,n+1)$ where $B(x,a,b)$ is the regularized incomplete beta-function. As can be seen in \cref{fig:schedule}, directly using the beta function as schedule, it deviates significantly from the reference $f(x)=x$ at intermediate times and the derivative $\partial_x\beta_n(x)=\frac{(1-x)^nx^n}{\int_0^1 y^n (1-y)^n dy}$ is maximal around $x=0.5$, and increases with $n$.\\

\subsection{Piecewise smooth schedule functions}

To ensure a more direct comparison between different $n$ we use a piecewise construction $s_{\beta_n,f}$, defined in \cref{eq:Ssigma} and illustrated in \cref{fig:schedule}, which is equal to $f(x)$ at intermediate times and smoothly transitions to a beta function with $n$ vanishing boundary derivatives only at the beginning and end.

Let $f(x):[0,1]\rightarrow[0,1]$ be a $C^n$-smooth transition function, and $\sigma(x):[0,1]\rightarrow[0,1]$ a $C^n$-smooth transition function with $n$ vanishing boundary derivatives. Then
\begin{equation}\label{eq:Ssigma}
     s_{\sigma,f}:x\mapsto\begin{cases}
         \sigma\left(\frac{x}{d}\right)f(x)  \quad &\hspace{-3ex} x\leq d\\
         f(x) \quad &\hspace{-7.3ex} d\leq x\leq 1-d \\
         \sigma\left(\frac{(x-1+d)}{d}\right)  &\hspace{-3ex} x\geq 1-d \\
         \quad +\left(1-\sigma\left(\frac{(x-1+d)}{d}\right)\right)f(x)
        \end{cases}
\end{equation}
is also $C^n$-smooth, and has $n$ vanishing boundary derivatives. Here $d<0.5$ is an adjustable parameter controlling how much time is spent on smoothing. In this work, we always use $\sigma(x)=\beta_n(x)$ as the smoothening function.

We also consider a schedule with diverging boundary derivatives $\textrm{Sqrt}(x)$. Note that there is no theorem for the scaling with diverging boundary derivatives, but one would expect diverging boundary derivatives to perform worse, and this is indeed what we observe. On a noisy device, the higher infidelities caused by the $\textrm{Sqrt}(x)$ schedule stand out more from the noise and therefore functions as an easier control experiment.
The $\textrm{Sqrt}$ schedule function with diverging boundary derivatives is defined by
  \begin{equation}\label{eq:Sqrt}
     x\mapsto\begin{cases}
         \left(1-\sigma_{\infty}\left(\frac{x}{2d}\right)\right)d\sqrt{\frac{x}{d}} \quad & x\leq 2d  \\
         \quad + \sigma_{\infty}\left(\frac{x}{2d}\right)f(x)   \\
         f(x) \quad &\hspace{-4.3ex} d\leq x\leq 1-2d \\
         \sigma_{\infty}\left(\frac{x-1+2d}{2d}\right)\left(1-d\sqrt{\frac{1-x}{d}}\right) \quad  & x\geq 1-2d \\
         \quad + \left(1-\sigma_{\infty}\left(\frac{x-1+2d}{2d}\right)\right)f(x)
         
     \end{cases}
 \end{equation}
 where $\sigma_{\infty}(x)$ is a $C^{\infty}$-smooth transition function.
 This ensures that there are no discontinuities in the derivatives at intermediate times, although $\partial_x\textrm{Sqrt}(x)$ diverges at $\tau = 0, 1$.

\subsection{Rydberg reference schedule}
For the Rydberg chain, we use a reference schedule that slows down near the minimum many-body gap. Let $\Delta(x)$ denote the gap between the ground state and first excited state of $H(x)$ along the path in parameter space. The non-adiabatic couplings for a schedule $s(x)$ are proportional to $\partial_x s(x)/\Delta(s(x))$. We impose
\begin{equation}\label{eq:igap_schedule}
    \partial_x s(x)/\Delta(s(x))=cg(x)
\end{equation}
where $g(x) \geq 0$ is a given function and $c$ a normalization constant to be determined. Given the gap $\Delta(x)$ and a target function $g(x)$, we can solve this equation for the schedule $s(x)$. We computed the gap numerically with DMRG for a length 11 chain, and integrated the equation numerically. We used a truncated cosine 
\begin{equation}\label{eq:IGapICosP05}
    g(x)=(1.5\pi x+\cos(\pi x)-1)/(1.5\pi-2), 
\end{equation}
so the schedule slows down in the middle. To ensure that the schedule is $C^n$-smooth we fitted a degree 10 polynomial to the numerical values and used this polynomial as the schedule. This schedule is suboptimal for long times $T\geq 4\microsec$ but, crucially for us, works well for short times.

\section{Proof of \texorpdfstring{\cref{thm:adiabatic derivatives}}{Theorem 1}}\label{sec:proof}
The proof of \cref{thm:adiabatic derivatives} is a modification of arguments made in \cite{hagedorn1989adiabatic,hagedorn2002elementary,lidar2009adiabatic}.
We first describe the set-up, and state the results on the superadiabatic expansion we need, following \cite{lidar2009adiabatic}.
Let $H(\tau)$ be a family of $k$ times differentiable self-adjoint operators on a separable Hilbert space, and denote by $H^{(k)}(\tau)$ its $k$th derivative in $\tau$. We let $\ket{\Phi(\tau)}$ be a non-degenerate eigenvector of $H(\tau)$ with a continuous curve of energies $E(\tau)$, which we assume to be separated from the rest of the spectrum by a distance $d(\tau)\geq d_0>0$ (i.e. this is the spectral gap in case of the ground state). Let $\ket{\psi(\tau,\eps)}$ be the solution to the Schr\"odinger equation
\begin{equation}\label{eq:schroedinger}
    i\eps\frac{\partial\ket{\psi(\tau,\eps)}}{\partial \tau}=H(\tau)\ket{\psi(\tau,\eps)}
\end{equation}
with initial condition
\begin{equation}\label{eq:adiabatic_ic}
    \ket{\psi(0,\eps)}=\ket{\Phi(0)}.
\end{equation}
This corresponds to a time evolution along $H(t/T)$ for time $t = T\tau$ with $\eps = T^{-1}$.
Let $U(\tau_2,\tau_1,\eps)$ denote the unitary propagator of the Schr\"odinger operator for \cref{eq:schroedinger} from time $\tau_1$ to $\tau_2$.
The adiabatic distance is defined by the norm difference
\begin{equation} \label{eq:distance}
\delta(\tau)=\Vert\ket{\psi(\tau,\eps)}-e^{-i\int_0^\tau E(\tau')d\tau'/\eps}\ket{\Phi(\tau)}\Vert \, .
\end{equation}

We now review some facts about the superadiabatic expansion.
The superadiabatic expansion solves the Schr\"odinger equation order by order in $\eps$.
By doing this to order $n$, this constructs an approximation $\ket{\Psi_n(\tau,\eps)}$ to the state $\ket{\psi(\tau,\eps)}$ given by
\begin{equation}\label{eq:superadiabatic}
\begin{split}
    &e^{i\int_0^\tau E(\tau')d\tau'/\eps} \ket{\Psi_n(\tau,\eps)}  \\
    &= \theta(\tau) \ket{\Phi(\tau)}+\sum_{k=1}^{n+1}\eps^{k}\ket{\psi_k^{\perp}(\tau)}
\end{split}
\end{equation}
where the $\ket{\psi_k^{\perp}(\tau)}$ are orthogonal to $\ket{\Phi(\tau)}$.
$\ket{\Psi_n(\tau,\eps)}$ is a solution to an approximate Schr\"odinger equation
\begin{equation}\label{eq:schroedinger_zeta}
    i\eps\frac{\partial\ket{\Psi_n(\tau,\eps)}}{\partial\tau}-H\ket{\Psi_n(\tau,\eps)} = \zeta_n(\tau,\eps)  \, ,  \\
\end{equation}
where
\begin{equation}\label{eq:zeta}
    \zeta_n(\tau,\eps) = i\eps^{n+2}e^{-i\int_0^\tau E(\tau')d\tau'/\eps}\frac{\partial\ket{{\psi}_{n+1}^\perp(\tau)}}{\partial\tau}
\end{equation}
is an error that comes from truncating the superadiabatic expansion at order $n+1$.
The functions $\theta(\tau)$ and $\ket{\psi_j^\perp(\tau)}$ can be constructed iteratively.
The functions $\theta(\tau)$ are determined up to an integration constant, which can be chosen freely. We choose them such that $\theta(0)=1$.
Note that $\ket{\Psi_n(\tau,\eps)}$ is not necessarily normalized.
We need two facts about the superadiabatic expansion in \cref{eq:superadiabatic}.

Firstly, if the derivatives of $H(\tau)$ vanish up to some order at some point, then the states $\ket{\psi_j^\perp(\tau_1)}$ of the corresponding orders vanish.
\begin{lem}[Lemma 1 of \cite{lidar2009adiabatic}]\label{lem:vanishing derivatives}
    If all derivatives $H^{(k)}(\tau_1)=0$ for all $1\leq k\leq n$ and some $\tau_1\in [0,1]$, then 
\begin{equation}
    \ket{\psi_j^\perp(\tau_1)}=0 \quad\textrm{for}\quad j\in\{1,...,n\}\, .
\end{equation}
\end{lem}
We also use the following basic estimate.
\begin{lem}[Generalization of Lemma 2.1 from \cite{hagedorn1989adiabatic}]\label{lem:error bound superadiabatic}
    If $\ket{\Psi_n(\tau,\eps)}$ solves the approximate Schr\"odinger equation \eqref{eq:schroedinger_zeta} and $\ket{\psi(\tau,\eps)}$ is a solution to the exact Schr\"odinger equation \eqref{eq:schroedinger}, then 
\begin{equation}\label{eq:lemma21_bound}
\begin{split}
    \Vert\ket{\psi(\tau,\eps)}-&\ket{\Psi_n(\tau,\eps)}\Vert \leq \int_0^\tau\eps^{-1}\Vert\zeta_n(s,\eps)\Vert ds \\
    &+ \norm{\ket{\psi(0,\eps)}-\ket{\Psi_n(0,\eps)}},
\end{split}
\end{equation}
for $\tau\in[0,1]$.
\end{lem}
\begin{proof} First we use unitarity of $U(\tau,0,\epsilon)$ and a triangle inequality to bound the LHS as follows
\begin{equation}\label{eq:lemma21_proof1}
\begin{split}
    &\norm{U(\tau,0,\epsilon)\ket{\psi(0,\eps)} - \ket{\Psi_n(\tau,\eps)}} \\
    &=\norm{\ket{\psi(0,\eps)} - U(0,\tau,\epsilon)\ket{\Psi_n(\tau,\eps)}} \\
    &\leq\norm{\ket{\Psi_n(0,\eps)} - U(0,\tau,\epsilon)\ket{\Psi_n(\tau,\eps)}} \\
    & + \norm{\ket{\psi(0,\eps)}-\ket{\Psi_n(0,\eps)}} .
\end{split}
\end{equation}
The remainder of the proof is identical to that in \cite{hagedorn1989adiabatic}.
\end{proof}

We now prove \cref{thm:adiabatic derivatives} by an argument similar to that in \cite{lidar2009adiabatic}.
First we note that since we assume that $H^{(k)}(0) = H^{(k)}(1) = 0$ for $1 \leq k \leq n$, we get from \cref{lem:vanishing derivatives} and \cref{eq:superadiabatic} that
\begin{equation}\label{eq:Psi at 0}
    \begin{split}
\ket{\Psi_n(0,\eps)} 
    &=\theta(0)\ket{\Phi(0)}+\eps^{n+1}\ket{\psi_{n+1}^\perp(0)}\\
    &= \ket{\Phi(0)}+\eps^{n+1}\ket{\psi_{n+1}^\perp(0)} \, ,
    \end{split}
\end{equation}
since we chose $\theta(0) = 1$ above, and
\begin{equation}\label{eq:Psi at 1}
    \begin{split}
        &e^{i\int_0^1 E(\tau')d\tau'/\eps} \ket{\Psi_n(1,\eps)} \\ 
    &=\theta(1)\ket{\Phi(1)}+\eps^{n+1}\ket{\psi_{n+1}^\perp(1)} \, .
    \end{split}
\end{equation}
We define the following error terms:
\begin{equation}
    \begin{split}
        \delta_0 &:= \eps^{n+1} \norm{\ket{\psi_{n+1}^\perp(0)}} \\
        \delta_1 &:= \eps^{n+1} \norm{\ket{\psi_{n+1}^\perp(1)}} \\
        \eta &:= \eps^{-1}\int_0^1 \Vert\zeta_n(s,\eps)\Vert ds  \\
        &=  \eps^{n+1}\int_0^1 \left\Vert \frac{\partial\ket{\psi_{n+1}^\perp(\tau')}}{\partial \tau'} \right\Vert d\tau' \, ,
    \end{split}
\end{equation}
from \cref{eq:zeta}. 
Each of these is $O(\eps^{n+1})$.

The states $\ket{\Psi_n(\tau,\eps)}$ are not normalized, but by \cref{eq:Psi at 0},
\begin{equation}\label{eq:norm Psi at 0}
   \abs{\norm{\ket{\Psi_n(0,\eps)}} -1 } \leq \delta_0 \, .
\end{equation}
Additionally, by Lemma 3 and \cref{eq:Psi at 0},
\begin{equation}
\begin{split}\label{eq:superadiabatic-exact}
    &\norm{\ket{\Psi_n(1,\epsilon)}-\ket{\psi(1,\epsilon)}} \\
    &\leq \eta + \norm{\ket{\Psi_n(0,\epsilon)}-\ket{\psi(0,\epsilon)}} \\
    &= \eta + \norm{\ket{\Psi_n(0,\epsilon)}-\ket{\Phi(0)}} \\
    &\leq \eta + \delta_0 =O(\epsilon^{n+1}).
\end{split}
\end{equation}
Since $\ket{\psi(\tau,\epsilon)}$ is normalized, it follows that the states $\ket{\Psi_n(\tau,\eps)}$ for $\tau = 0,1$ are normalized up to terms of order $\eps^{n+1}$. In particular, 
\begin{equation}\label{eq:theta 1 close to 1}
\begin{split}
    \abs{\theta(1)} &= \norm{e^{i\int_0^1 E(\tau')d\tau'/\eps}\ket{\Psi_n(1,\eps)} - \eps^{n+1}\ket{\psi_n^{\perp}(1,\eps}} \\
    &\geq 1 - (\eta + \delta_0 + \delta_1) \\
        &= 1 - O(\eps^{n+1}) \, 
\end{split}
\end{equation}
follows from \cref{eq:Psi at 1}, \cref{eq:superadiabatic-exact} and the definition of $\delta_1$.
Next we compute
\begin{equation}\label{eq:superadiabitic-ground}
\begin{split}
    \norm{&\ket{\Psi_n(1,\epsilon)}-e^{-i\int_0^1 E(\tau')d\tau'/\eps}\ket{\Phi(1)}} \\
    &\leq\norm{e^{-i\int_0^1 E(\tau')d\tau'/\eps}(\theta(1)-1)\ket{\Phi(1)}} \\
    & \quad +\norm{\epsilon^{n+1}\ket{\psi^\perp_{n+1}(1)}} \\
    &=\abs{\theta(1)-1}+\delta_1 \\
    &\leq \eta+\delta_0+2\delta_1 = O(\eps^{n+1})
\end{split}
\end{equation}
using \cref{eq:Psi at 1} in the first inequality and \cref{eq:theta 1 close to 1} in the second one.
Bringing everything together, we get
\begin{equation}\label{eq:conclusion}
\begin{split}
    \delta(1)&=\norm{\ket{\psi(1,\epsilon)}-e^{-i\int_0^1 E(\tau')d\tau'/\eps}\ket{\Phi(1)}} \\
    &\leq \norm{\ket{\psi(1,\epsilon)}-\ket{\Psi_n(1,\epsilon)}} \\
    &\quad + \norm{\ket{\Psi_n(1,\epsilon)}-e^{-i\int_0^1 E(\tau')d\tau'/\eps}\ket{\Phi(1)}}  \\
    &\leq 2\eta+2\delta_0+2\delta_1 =O(\eps^{n+1})
\end{split}
\end{equation}
using \cref{eq:superadiabatic-exact} and \cref{eq:superadiabitic-ground} in the last line.

\section{Estimation of first-order component}\label{sec:first-order}
Consider the case where $\ket{\Phi(\tau)} = \ket{\Phi_0(\tau)}$ is the ground state.
If $\ket{\Phi_n(\tau)}$ denote the higher energy eigenstates with energies $E_n(\tau)$, then the 
the first-order leakage into $\ket{\Phi_n(\tau)}$ is \cite{benseny2021adiabatic}
\begin{equation}
    \delta^1_{0,n}(\tau) = \eps\frac{\bra{\Phi_n(\tau)}\partial_\tau \ket{\Phi_0(\tau)}}{\Delta_{n,0}(\tau)}
\end{equation}
where $\Delta_{k,n}(\tau) = E_k(\tau)-E_n(\tau)$. Most leakage occurs to the first excited state, so instead of computing all couplings we approximate them by $\eps\gamma_0(\tau)/\Delta_{1,0}(\tau)$ where
\begin{align}\label{eq:gamma_gap}
\begin{split}
    \gamma_0(\tau) &= \sqrt{1-|\braket{\Phi_0(\tau) |\partial_\tau| \Phi_0(\tau)}|^2} \\
    &\approx \sqrt{1-|\braket{\Phi_0(\tau) | \Phi_0(\tau+d\tau)}/d\tau|^2}  \, .
    \end{split}
\end{align}
Note that this is a conservative estimate because the higher nonadiabatic coupling terms are included in $\gamma_0$ but are not suppressed by larger energy gaps $\Delta_{k,0}$.

\end{document}